\documentclass[12pt,a4paper]{article}
\usepackage{amsmath,amsfonts,amssymb,amsthm,epsfig,epstopdf,titling,url,array}
\usepackage{geometry}
\usepackage{amsmath}
\usepackage{float}
\usepackage[english]{babel}
\usepackage{enumitem}
\usepackage{caption}
\usepackage{csquotes}
\usepackage[margin=8pt,font=small,labelfont=bf]{caption}
\usepackage[colorlinks,
            linkcolor=blue,
            urlcolor=blue,
            anchorcolor=blue,
            citecolor=blue
            ]{hyperref}
\usepackage{color}
\usepackage{xcolor}
\usepackage[normalem]{ulem}
\usepackage{tikz}
\usetikzlibrary{arrows,automata,positioning}

\theoremstyle{plain}

\newtheorem{theorem}{Theorem}[section]

\newtheorem{example}[theorem]{Example}
\newtheorem{assumption}[theorem]{Assumption}
\newtheorem{remark}[theorem]{Remark}
\newtheorem{algorithm}[theorem]{Algorithm}
\newtheorem{proposition}[theorem]{Proposition}

\theoremstyle{remark}

\DeclareUnicodeCharacter{1D0C}{*************************************}
\title{Price-mediated contagion with endogenous market liquidity}
\author{Zhiyu Cao\thanks{Stevens Institute of Technology, School of Business, Hoboken, NJ 07030, USA} 
\and 
Zachary Feinstein\thanks{Stevens Institute of Technology, School of Business, Hoboken, NJ 07030, USA, {\tt zfeinste@stevens.edu}}}

\date{\today}

\usepackage{graphicx}
\usepackage{subcaption}
\usepackage[toc,page]{appendix}

\usepackage[nottoc]{tocbibind}

\newcommand{\M}{\mathcal{M}}

\begin{document}

\maketitle

\begin{abstract}
Price-mediated contagion occurs when a positive feedback loop develops following a drop in asset prices which forces banks and other financial institutions to sell their holdings. Prior studies of such events fix the level of market liquidity without regards to the level of stress applied to the system. This paper introduces a framework to understand price-mediated contagion in a system where the capacity of the market to absorb liquidated assets is determined endogenously. In doing so, we construct a joint clearing system in interbank payments, asset prices, and market liquidity. We establish mild assumptions which guarantee the existence of greatest and least clearing solutions. We conclude with detailed numerical case studies which demonstrate the, potentially severe, repercussions of endogenizing the market liquidity on system risk.
\end{abstract}

\section{Introduction}

The global financial system is an intricate network of interlinked financial institutions. These interconnections occur through, e.g., debt obligations, cross-ownership, and common portfolio holdings. Due to possible network effects, adverse actions by a single bank can distress other entities, potentially initiating a cascade of severe economic repercussions across the financial system. This contagion effect, also called ``systemic risk'', can lead to excessive harm to the global economy. As shown in the 2007--2009 financial crisis, Lehman Brothers' size and interconnectedness with other financial institutions made it a source of systemic risk. Though prior works have analyzed the potential for default and price-mediated contagion due to distress of firms like Lehman Brothers, we include further deterioration in financial health because of a drop in the market liquidity provided by these distressed institutions. Specifically, this paper extends the financial contagion model from \cite{feinstein2017financial} to study the impact of bank distress on market liquidity. In light of their critical role in maintaining liquidity in the financial system, we will frequently characterize healthy banks (which provide the market liquidity) as market makers.

The organization of this paper is as follows. We complete this introduction by providing a review of the relevant literature in Section~\ref{sec:intro-lit} and a summarization of the primary contributions of this work in Section~\ref{sec:intro-contrib}. In Section \ref{sec:setting}, we extend the framework of \cite{feinstein2017financial} to explicitly incorporate market liquidity in price formation. Section \ref{sec:clearing} provides the primary theoretical results of this work. Specifically, with Section \ref{sec:clearing-model} and \ref{sec:clearing-nonunique}, we prove the existence of greatest and least clearing solutions and demonstrate these are (generally) not equivalent. 
We conclude with numerical case studies in Section \ref{sec:cs} to explore some simple financial implications of this liquidity-augmented contagion model.

\subsection{Literature Review}\label{sec:intro-lit}
Within this work, we study systemic risk via two primary contagion channels: default contagion through interbank obligations and price-mediated contagion through fire sales.

Default contagion arises through the network of counterparty relationships in the financial system. When one financial firm defaults on its obligations, its counterparties take losses on their balance sheets. The resulting write-downs may cause the distress and default of these, otherwise healthy, institutions causing a propagation of losses that cascades through this network of obligations. Such an equilibrium payment system was modeled in the seminal work of \cite{eisenberg2001systemic}; within that work the existence and uniqueness of these clearing payments was determined under simple and natural technical conditions. This model has been extended in many directions. For instance, \cite{acemoglu2015systemic} showed the `robust-yet-fragile' tendency of financial networks topology, \cite{rogers2013failure} introduced bankruptcy costs to this framework and \cite{ghamami2022collateralized} proposed an extension to allow for collateralized obligations. More recently, \cite{veraart2020distress,feinstein2022endogenous} allow for counterparty risk to propagate before defaults have been realized by the system through distress contagion. The aforementioned default contagion models are presented in a single time step setting, time-dynamic systems have been introduced in, e.g., \cite{capponi2015systemic,banerjee2018dynamic,feinstein2020capital,sonin2020continuous}.

Price-mediated contagion is largely driven by the use of mark-to-market accounting rules. When the price of assets on a firm's balance sheet drops, its liquidity can evaporate triggering it to sell assets in order to recapitalize. However, selling assets in a depressed market can cause further deterioration in asset prices; such price drops can trigger more firms to liquidate assets further spreading this distress throughout the financial system. This form of contagion has been studied under multiple different regulatory regimes. \cite{cifuentes2005liquidity,amini2016uniqueness,feinstein2017financial,bichuch2019optimization} incorporate fire sales within the Eisenberg-Noe default contagion system \cite{eisenberg2001systemic}.
Separately, price-mediated contagion was studied within leverage and capital ratio-constrained systems within, e.g., \cite{greenwood2015vulnerable,braouezec2019strategic,cont2019monitoring}.

To model the price impacts from liquidating assets, the price-mediated contagion literature primarily utilizes inverse demand functions. These functions map the number of sold assets into the market clearing price. Linear inverse demand functions were presented within, e.g., \cite{greenwood2015vulnerable,braouezec2019strategic}, and exponential inverse demand functions were utilized in, e.g., \cite{cifuentes2005liquidity}.  Generalizations of these forms were presented in, e.g., \cite{amini2016uniqueness,feinstein2017financial}, and a data-driven model for the inverse demand function was introduced in \cite{cao2023modeling}.
Distinct from these approaches, \cite{bichuch2022endogenous} proposes an explicit market clearing construction for the inverse demand function assuming all market participants are utility maximizers within an exchange economy.

\subsection{Primary Contribution}\label{sec:intro-contrib}

As highlighted above, within this work we introduce an extension of \cite{feinstein2017financial} that incorporates both default and price-mediated contagion to, additionally, directly model the market liquidity available during a stress scenario. Specifically, herein, banks with the ability to provide liquidity are treated as market makers. The distress of market makers can amplify price-mediated contagion by increasing the price impacts from selling assets. Leveraging the inverse demand function framework of \cite{bichuch2022endogenous}, we consider an equilibrium inverse demand function that explicitly accounts for market liquidity in price formation.

Though the framework herein is a combination of the contagion model of \cite{feinstein2017financial} with the inverse demand structure of \cite{bichuch2022endogenous}, we find that the joint effects lead to additional mathematical complexities and financial impacts. We view our primary innovations thusly:
\begin{enumerate}
\item Though the constructions taken herein have been considered independently in the literature, the joint effects have not previously been studied. In particular, directly studying market liquidity within price-mediated contagion is, as far as the authors are aware, novel to the literature. Existing literature (e.g., \cite{amini2015systemic,feinstein2017financial,weber2017joint,banerjee2020price}) finds the clearing prices under an assumption of constant market liquidity, i.e., so that the equilibrium prices are independent of the stress scenario impacting bank balance sheets. However, as observed in the 2008 financial crisis, the lack of liquidity amplifies the price impacts from market orders. 
In proposing this extended model, we quantify the losses caused by endogenizing the market liquidity compared to the fixed (exogenous) liquidity setting. 
\item As discussed, our framework combines the contagion model of \cite{feinstein2017financial} and the inverse demand functions of \cite{bichuch2022endogenous}. However, though considering either system in isolation results in unique solutions, the joint model can result in the non-uniqueness of the clearing solutions. This non-uniqueness is demonstrated in a small case study (Section~\ref{sec:clearing-nonunique}), though existence of greatest and least clearing is guaranteed under mild technical assumptions (Section~\ref{sec:clearing-model}). This emergent behavior of the joint model demonstrates the necessity to consider joint models for systemic risk as the individual components alone can neglect such challenges.
\item Similar to \cite{feinstein2019obligations}, the non-uniqueness of the clearing solutions can lead to downward jumps in the equilibrium health of the financial system. In fact, as this cannot be captured by sensitivity analysis, these events exacerbate systemic risk as simply considering a discrete number of stress scenarios need not closely match the realized behavior. As such, these results underscore the importance of systematically analyzing the system for these potential points of discontinuity to verify the health of the financial network. These discontinuities are explored in numerical case studies.
\end{enumerate}

\paragraph*{Notation:} Fix $x,y \in \mathbb{R}^n$ for some positive integer $n$. Throughout this work, we will compare these vectors component-wise, i.e., $x \leq y$ if and only if $x_i \leq y_i$ for every $i \in \{1,2,\ldots,n\}$. Furthermore, if $x \leq y$, we define the $n$-dimensional box $[x,y] = [x_1,y_1] \times [x_2,y_2] \times \ldots \times [x_n,y_n]$. Additionally, throughout this work, we define the lattice minimum and maximum, respectively, as:
\begin{align*}
x \wedge y&=\left(\min \left(x_1, y_1\right), \min \left(x_2, y_2\right), \ldots, \min \left(x_n, y_n\right)\right)^{\top}, \\
x \vee y&=\left(\max \left(x_1, y_1\right), \max \left(x_2, y_2\right), \ldots, \max \left(x_n, y_n\right)\right)^{\top}.
\end{align*}
The positive part of a vector is denoted by $x^+ = x \vee 0$ and the negative party by $x^- = (-x)^+.$
Finally, we denote the positive orthant in $n$-dimensional space as $\mathbb{R}^n_+ := \{x \in \mathbb{R}^n \; | \; x \geq 0\}$.

\section{Financial Setting}\label{sec:setting}

Within this work, we consider a banking system with \( n \) financial institutions (e.g., banks, hedge funds, and pension plans) and a financial market with \( m \) illiquid assets. Throughout this paper, we fix notation so that \( p \in \mathbb{R}^{n}_{+} \) represents the payments made by these institutions, \( q \in \mathbb{R}_{+}^{m} \) represents the prices of the illiquid assets, and \( M \in \mathbb{R}^n_+ \) to denotes the market liquidity provided by each financial institution.

\subsection{Balance sheet model and payments}\label{sec:setting-payments}
Drawing from the framework presented in \cite{feinstein2017financial}, each bank $i$ in our model possesses three categories of assets: (i) liquid assets \(x_{i} \geq 0\); (ii) illiquid assets \(s_{ik} \geq 0\) units of each asset $k \in \{1,2,\ldots,m\}$; and (iii) interbank assets \(L_{ji} \geq 0\) of obligations owed by bank $j \in \{1,2,\ldots,n\}$.\footnote{We assume that no bank holds any liabilities to itself, i.e., $L_{ii} = 0$ for every bank $i$.} On the other side of the banking book, each bank $i$ has two types of liabilities: (i) external liabilities \(L_{i0} \geq 0\),\footnote{When explicitly modeled, the external economy will sometimes be called ``society'' or the ``societal node''.} and (ii) interbank liabilities \(\sum_{j = 1}^n L_{ij} \geq 0\). The capital of bank $i$ is the difference between the value of its assets and liabilities.
Often we will denote these balance sheet elements as vectors and matrices, i.e., $x \in \mathbb{R}^n_+$, $S \in \mathbb{R}^{n \times m}_+$ (with $s_i \in \mathbb{R}^m_+$ denoting the holdings for bank $i$ only), and $L \in \mathbb{R}^{n \times (n+1)}_+$ with total liabilities $\bar p = L\vec{1}$.

In constructing this banking book, all assets and liabilities were valued at their nominal values. However, under mark-to-market accounting rules, we assume that the price of the illiquid assets may fluctuate with price vector $q \in \mathbb{R}^m_+$ so that the aggregate value of bank $i$'s illiquid portfolio is \(\sum_{k = 1}^m s_{ik} q_k\). Additionally, we value the interbank assets via the Eisenberg-Noe clearing system \cite{eisenberg2001systemic} so that bank $i$ receives $p_{ji} \in [0, L_{ji}]$ from bank $j$. In particular, this clearing system assumes a pro-rata repayment scheme so that $p_{ji} = a_{ji} p_j$ such that $a_{ji} = L_{ji}/\bar{p}_j$ (if $\bar p_j > 0$ and $a_{ji} = 0$ otherwise) are the relative liabilities from bank $j$ to bank $i$ and $p_j \in [0,\bar p_j]$ is the total payments made by bank $j$.
Thus, the total value of bank $i$'s assets is given by $x_{i}+\sum_{k=1}^{m} s_{i k} q_{k}+\sum_{j=1}^{n} a_{j i} p_{j}$ under payment vector $p \in [0,\bar p]$ and price vector $q \in \mathbb{R}_{+}^{m}$.

Following from limited liabilities and the priority of debt over equity, the amount that bank $i$ pays into the financial system is the minimum of its total liabilities $\bar{p}_{i}$ and the value of its assets, i.e, the realized payment $p$ must satisfy the fixed point problem:
\begin{equation} \label{2.1}
p=\bar{p} \wedge\left(x+Sq + A^\top p\right)
\end{equation}
where $A \in [0,1]^{n \times n}$ encodes the relative liabilities.
In Section~\ref{sec:clearing-model} below, we will determine the existence and uniqueness properties of this payment system. We note that, neglecting the impacts of liquidating assets on prices $q$, the properties of this payment system were fully characterized within \cite{eisenberg2001systemic}.

\subsection{Pricing illiquid assets}\label{sec:setting-prices}
Within the prior discussion, we have treated the prices \(q \in \mathbb{R}^m_+\) of the illiquid assets as fixed. Herein, we want to characterize these prices as a function of the quantity of illiquid assets sold in the market. 

\subsubsection{Asset liquidations}
Following \cite{feinstein2017financial}, we denote the number of physical units of asset $k \in \{1,2,\ldots,m\}$ that bank $i \in \{1,2,\ldots,n\}$ wants to sell by the liquidation function $\gamma_{ik}: [0,\bar p] \times \mathbb{R}^m_+ \to [0,s_{ik}]$. (Often, we will denote the vector of liquidations by bank $i$ by $\gamma_i$ without reference to the asset.) These liquidation functions map the payments ($p \in [0,\bar p]$) and prices ($q \in \mathbb{R}^m_+$) into the number of assets to be sold. These assets are being sold by each bank so as to cover any obligations it is unable to meet through its liquid and interbank assets alone. Implicit to the construction of the liquidation function, we assume that no banks acquire additional assets during the liquidation process ($\gamma_i(p,q) \geq 0$)
 nor do they sell more than they own ($\gamma_i(p,q) \leq s_i$).

Throughout this work we will impose the following minimal liquidation condition (see also, e.g., \cite{feinstein2017financial,banerjee2019impact}). This condition implies that all banks will never sell more assets than necessary to remain solvent while also totally liquidating their position when entering into default.  From the bank's perspective, since liquidating illiquid assets impacts their prices, utility maximization dictates that banks want to liquidate the fewest assets necessary to meet their obligations.
\begin{assumption}\label{Assumption 2.1}
For each bank $i$, the liquidation function $\gamma_i: [0, \bar{p}] \times \mathbb{R}^m_+ \to [0,s_i]$ satisfies the \emph{minimal liquidation condition}: 
\begin{align*}
&q^{\top} \gamma_{i}(p, q) = \left(q^{\top} s_{i}\right) \wedge\left(\bar{p}_{i}-x_{i}-\sum_{j=1}^{n} a_{ji} p_{j}\right)^{+}
\end{align*}
for every $p \in [0,\bar p]$ and $q \in \mathbb{R}^m_+$.
\end{assumption}

Assumption \ref{Assumption 2.1} states that the amount liquidated by a bank is either enough to cover the liquid shortfall in obligations or all of its assets are liquidated. The minimal liquidation condition also implies that given higher payments from other banks and higher asset prices, banks will liquidate (in aggregate) fewer assets to cover their shortfall. This latter property is strengthened and formalized within the next assumption.

\begin{assumption}
For each bank $i$, the liquidation function $\gamma_i: [0, \bar{p}] \times \mathbb{R}^m_+ \rightarrow [0,s_i]$ is monotonically non-increasing (i.e., $ \gamma_i(p^{1}, q^{1}) \leq \gamma_i(p^{2}, q^{2})$ for any payments $p^1 \geq p^2$ and prices $q^1 \geq q^2$).
\end{assumption}

Before continuing our discussion where we model the price impacts, we wish to provide an example of a liquidation function. Specifically, we characterize the behavior of a bank that sells its assets in proportion to its holdings, i.e., that sells units of its whole portfolio rather than rebalancing that portfolio.

\begin{example}\label{example2.2.3}
Following \cite{cont2017fire, greenwood2015vulnerable, duarte2021fire}, consider the rule where banks liquidate their assets in proportion to their initial holding. Specifically, for each bank \(i\), and assets \(k, l=1,2, \ldots, m\):
\[\frac{\gamma_{il}(p, q)}{\gamma_{i k}(p, q)}=\frac{s_{il}}{s_{i k}}\]
for any payment vector $p \in [0,\bar p]$ and price vector $q \in \mathbb{R}^m_+$.
In fact, the minimal liquidation condition of Assumption \ref{Assumption 2.1} fully characterizes these proportional liquidations. Specifically, for any bank $i$ and asset $k$, the liquidations are given by
\[\gamma_{ik}(p,q) = \frac{s_{ik}}{q^\top s_i} \left[(q^\top s_i) \wedge \left(\bar p_i - x_i - \sum_{j = 1}^n a_{ji} p_j\right)^+\right] \]
for any $p \in [0,\bar p]$ and $q \in \mathbb{R}^m_+$ such that $q^\top s_i > 0$.\footnote{If $q^\top s_i = 0$ then the liquidation strategy $\gamma_i$ is irrelevant as the bank is unable to raise any funds from selling any portion of its portfolio.}
\end{example}

\subsubsection{Inverse demand function}
As banks implement their liquidation strategies ($\gamma_i(p,q)$ for bank $i$), the prices do not remain fixed. Rather, the price of the illiquid assets are directly dependent on the total number of each asset being sold due to price impacts.  
However, motivated by \cite{bichuch2022endogenous} and in contrast to, e.g., \cite{braouezec2019strategic, cifuentes2005liquidity, weber2017joint, feinstein2017financial,banerjee2020price}, the scale of these price impacts depends also on the amount of liquidity available to the market. For each bank, it has excess liquidity equal to the difference between its liquid holdings $x_i + \sum_{j = 1}^n a_{ji} p_j$ and its obligations $\bar p_i$ if that is positive and $0$ if negative, i.e., we define the vector of market liquidities by 
\begin{equation}\label{eq.2.2}
M = \left(x + A^\top p - \bar p\right)^+
\end{equation}
for payment vector $p \in [0,\bar p]$.
This excess liquidity $M_i$ is used, implicitly, to rebalance the portfolio of bank $i$ when it is not liquidity constrained (i.e., $M_i > 0$). In particular, these liquid banks are able to purchase assets at a discount during a fire sale of $\sum_{i = 1}^n \gamma_i(p,q)$ assets. Therefore, generally, if there is more liquidity in the market, the asset prices will be more stable, i.e., with smaller price impacts. However, if there is a lack of market liquidity, the illiquid assets will quickly depreciate in value.

\begin{remark}\label{remark2.4}
Throughout this work, we often refer to banks that are able to provide market liquidity as market makers. Specifically, the set \(\mathcal{M} := \{i \in \{1,2,\ldots,n\} \; | \; M_{i} > 0\}\) represents the market makers in the system. 
That is, the set of market makers coincides exactly with those institutions that do not liquidate any assets so that they have a liquidity surplus. 
The market liquidity comes directly from these market makers; therefore, as market makers enter distress, the price impacts realized from liquidations increase.
Notably, as a bank's liquidity $M_i$ explicitly depends on the clearing payments $p$, the set of market makers is variable with respect to the health of the system; this is in contrast to the classical approach to inverse demand functions in which the price impacts are independent of the clearing solutions (see, e.g.,~\cite{cifuentes2005liquidity,greenwood2015vulnerable,amini2016uniqueness}).
In fact, it is quite likely that many institutions that look to have book liquidity shortfalls ($x_i < \bar p_i$) can act as market makers based on the payments they receive from other institutions. Therefore, as the health of the system deteriorates, the overlap between those institutions with book liquidity shortfalls and market makers shrinks. We display a stylized categorization of institutions in Figure~\ref{Venn}; we note that, other than the set of institutions with book liquidity shortfall, the classification of any bank depends on the health of the system and may deteriorate during the clearing procedure.
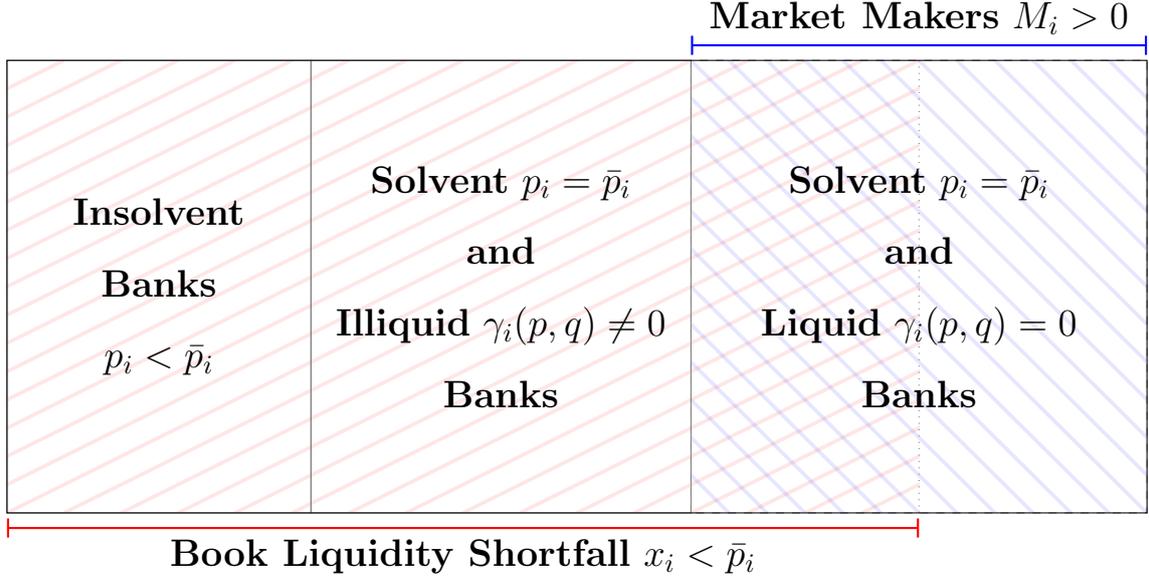
\begin{figure}
\centering
\begin{tikzpicture}
\begin{scope}
    \clip (0,0) rectangle (12,6);
    \foreach \x in {-12,-11.25,...,12}
    {
        \draw[line width=.5mm,draw=red,opacity=0.1] (\x,0) -- (12+\x,6);
    }
\end{scope}
\begin{scope}
    \clip (9,0) rectangle (15,6);
    \foreach \y in {-8,-7.5,...,8}
    {
        \draw[line width=.5mm,draw=blue,opacity=0.1] (9,6+\y) -- (15,\y);
    }
\end{scope}
\draw [draw=gray,dotted] (12,0) -- (12,6);
\draw [draw=gray] (0,0) rectangle (4,6) node[font=\large,pos=0.5,align=center] {\textbf{Insolvent} \\ \textbf{Banks} \\ $p_i < \bar p_i$};
\draw [draw=gray] (4,0) rectangle (9,6) node[font=\large,pos=0.5,align=center] {\textbf{Solvent} $p_i = \bar p_i$ \\ \textbf{and} \\ \textbf{Illiquid} $\gamma_i(p,q) \neq 0$ \\ \textbf{Banks}};
\draw [draw=gray,dashed] (9,0) rectangle (15,6) node[font=\large,pos=0.5,align=center] {\textbf{Solvent} $p_i = \bar p_i$ \\ \textbf{and} \\ \textbf{Liquid} $\gamma_i(p,q) = 0$ \\ \textbf{Banks}};
\draw [|-|,thick,draw=blue] (9,6.2) -- (15,6.2) node[font=\large,anchor=south,pos=0.5,align=center] {\textbf{Market Makers} $M_i > 0$};
\draw [|-|,thick,draw=red] (0,-0.2) -- (12,-0.2) node[font=\large,anchor=north,pos=0.5,align=center] {\textbf{Book Liquidity Shortfall} $x_i < \bar p_i$};
\draw [draw=black] (0,0) rectangle (15,6);
\end{tikzpicture}
\caption{Stylized categorization of banks within our model. Note that the classification of any bank depends on the outcome of the clearing procedure.}
\label{Venn}
\end{figure}
\end{remark}

Consider a scenario in which the payments were $p \in [0,\bar p]$, the prices were quoted as $q \in \mathbb{R}^m_+$, and the vector of market liquidities was $M \in \mathbb{R}^n_+$. As discussed above, given this scenario, the entire financial system would seek to liquidate the portfolio $\sum_{i = 1}^n \gamma_i(p,q)$ of assets. Implementing these liquidations can cause a change in the price $q$ to 
\begin{equation}\label{2.3}
q' = F\left(\sum_{i = 1}^n \gamma_i(p,q),M\right)
\end{equation}
for \emph{inverse demand function} $F: \mathbb{R}^m_+ \times \mathbb{R}^n_+ \to \mathbb{R}^m_+$ mapping liquidations and market liquidity into resulting asset prices.

\begin{assumption}\label{assumption 2.3}
The inverse demand function $F: \mathbb{R}_{+}^{m} \times \mathbb{R}_{+}^{n} \to \mathbb{R}^m_+$ is monotonically non-increasing in the liquidations, i.e., the first argument ($F(\theta^{1}, M)\leq F(\theta^{2}, M)$ for $\theta^{1},\theta^2 \in \mathbb{R}^m_+$ such that $\theta^1 \geq \theta^{2}$) and monotonically non-decreasing in market liquidity, i.e., the second argument ($F(\theta, M^{1})\leq F(\theta, M^{2})$ for $M^1,M^2 \in \mathbb{R}^n_+$ such that $M^{1}  \leq M^{2}$). 
\end{assumption}

Under Assumption \ref{assumption 2.3}, when $M$ is fixed, the relationship \( \theta \in \mathbb{R}^m_+ \mapsto F(\theta, M) \) is monotonically non-increasing. This implies that as the supply of an illiquid asset in the market increases, its price cannot increase; in fact, due to the possibility of cross-impacts as encoded in this vector-valued inverse demand function, selling any assets cannot increase the value of any other assets either. This behavior aligns with traditional assumptions related to inverse demand functions. 

Prior literature has utilized inverse demand functions, without the dependence on market liquidity $M$, with linear (see, e.g., \cite{greenwood2015vulnerable,cont2017fire}) and exponential (see, e.g., \cite{cifuentes2005liquidity}) forms.
The inclusion of the market liquidity in the inverse demand function is motivated by \cite{bichuch2022endogenous}. That work presents an equilibrium formulation for price impacts in an exchange economy to construct the inverse demand function. That is, illiquid assets are sold into a system of market participants. It is this exchange economy of unstressed firms (and, possibly, external liquidity) which is used to purchase the fire sold assets. Such a construction permits us to directly study the price impacts generated from the drying up of market liquidity on prices and, ultimately, systemic risk. 
Assumption \ref{assumption 2.3} implies that, for fixed liquidations, the more available market liquidity, the smaller the price impacts. However, the distress of the banks can leave only a few large institutions providing liquidity, which can trigger significant asset price depreciation.

We conclude this discussion of the inverse demand function with explicit dependence on market liquidity by providing an example of one such construction. This liquidity-adjusted linear inverse demand function is motivated by \cite[Example 4.4]{bichuch2022endogenous} and will be used throughout the numerical studies in the subsequent sections. 

\begin{example}\label{example2.1}
Consider a market with $n$ banks and $m$ assets. Assume each bank $i$ is an exponential utility maximizer with risk aversion $\alpha_i > 0$. In addition, assume that society behaves similarly with external risk aversion $\alpha_0 > 0$. Furthermore, assume the assets follow a joint normal distribution $N(\mu,C)$. 
As taken in \cite{bichuch2022endogenous}, this is a closed financial system, i.e., all liquidated assets need to be purchased by either one of the firms or society.
The resulting equilibrium inverse demand function from an exchange economy, as provided by \cite[Example 4.4]{bichuch2022endogenous}, extends the classical linear inverse demand function of, e.g., \cite{greenwood2015vulnerable} so that the price-impacts explicitly dependent on the set of market makers $\M = \{i \; | \; M_i > 0\}$. Specifically, for liquidations $\theta \in \mathbb{R}^m_+$ and market liquidity $M \in \mathbb{R}^n_+$, the inverse demand function is given by 
\[
F(\theta,M) = \mu - \left(\frac{1}{\alpha_0} + \sum_{i \in \M} \frac{1}{\alpha_i}\right)^{-1} C \theta.
\]
The price impacts of this liquidity-adjusted linear inverse demand function depend strongly on the risk aversions of society $\alpha_0$ and any market maker $\alpha_i$ for $i \in \M$. Therefore, we caution that any mis-calibration of these risk aversions can result in misspecification of the price impacts.
We wish to note that $F$ satisfies Assumption~\ref{assumption 2.3} if, and only if, all assets are non-negatively correlated, i.e., $C_{kl} \geq 0$ for every pair of assets $(k,l)$.
\end{example}

\section{Clearing payments, prices, and market liquidity}\label{sec:clearing}
In this section, we provide the primary theoretical results of this work. Specifically, in Section~\ref{sec:clearing-model}, we formalize the joint clearing problem in payments, prices, and market liquidity from the balance sheet constructions in Section~\ref{sec:setting} above. Therein we also prove the existence of greatest and least equilibria. In Section~\ref{sec:clearing-nonunique}, we provide a simple counterexample to the uniqueness of the clearing solution. This is in contrast to, e.g., \cite{amini2016uniqueness,feinstein2017financial} in which mild assumptions on the monotonicity of the cash raised from asset sales can guarantee uniqueness of the clearing solution. The implications of multiple equilibria are explored in more depth within the numerical case studies of Section~\ref{sec:cs}.

\subsection{Equilibrium formulation}\label{sec:clearing-model}
Recall the financial setting provided within Section~\ref{sec:setting} in which banks have external liquid assets $x \in \mathbb{R}^n_+$, illiquid assets $Sq \in \mathbb{R}^{n}_+$ at prices $q \in \mathbb{R}^m_+$, and interbank assets $A^\top p \in \mathbb{R}^{n}_+$ as well as (interbank and external) obligations $\bar p \in \mathbb{R}^n_+$. 
Given the interdependencies between payments, prices, and liquidity, the clearing system must account for all three objects simultaneously. These constructions were given previously in \eqref{2.1}, \eqref{2.3}, and \eqref{eq.2.2} respectively.  Formally, we define the clearing mechanism $\Phi: [0,\bar p] \times \mathbb{R}^m_+ \times \mathbb{R}^n_+ \to [0,\bar p] \times \mathbb{R}^m_+ \times \mathbb{R}^n_+$ by
\begin{equation} \label{phi_func}
\Phi(p, q, M)=\left(\begin{array}{c} 
\Phi_p(p, q, M)\\
\Phi_q(p, q, M)\\
\Phi_M(p, q, M)\\
\end{array}\right)
=\left(\begin{array}{c}
\bar{p}_{i} \wedge\left(x_{i}+\sum_{k=1}^{m}s_{ik} q+\sum_{j=1}^{n} a_{j i} \cdot p_{j}\right) \\
F\left(\sum_{i=1}^{n}\left[\gamma_{i}(p, q) \right], M\right)\\
(x_{i}+\sum_{j=1}^{n} a_{ji} {p_{j}}-\bar{p}_{i})^{+}
\end{array}\right)
\end{equation}
for any payments $p \in [0,\bar p]$, prices $q \in \mathbb{R}^m_+$, and market liquidity $M \in \mathbb{R}^n_+$.
Here, \(\Phi_p\) represents the updated payments, \(\Phi_q\) denotes the adjusted prices, and \(\Phi_M\) signifies the resulting market liquidity. The equilibrium payments, prices, and market liquidity are given by fixed points of $\Phi$, i.e.,
\begin{equation*} 
\left(p^{*}, q^{*}, M^{*}\right)=\Phi\left(p^{*}, q^{*},M^{*}\right)
\end{equation*}

\begin{theorem} \label{Theorem 3.2.1}
Consider the clearing mechanism $\Phi$ defined in \eqref{phi_func} where the inverse demand function $F$ satisfies Assumption~\ref{assumption 2.3}. Every clearing solution $(p^*,q^*,M^*) = \Phi(p^*,q^*,M^*)$ is an element of the lattice $[0,\bar p] \times [0,F(0,\bar M)] \times [0,\bar M]$ where $\bar M := (x + A^\top \bar p - \bar p)^+$. Furthermore, the set of clearing solutions form a complete lattice and, in particular, there exist a greatest and least greatest and least equilibria $\left(p^{\uparrow},q^{\uparrow}, M^{\uparrow}\right) \geq\left(p^{\downarrow},q^{\downarrow}, M^{\downarrow}\right)$. 
\end{theorem}
\begin{proof}
First, let $(p^*,q^*,M^*) \in [0,\bar p] \times \mathbb{R}^m_+ \times \mathbb{R}^n_+$ be a clearing solution. Then, since the relative liabilities $A \in [0,1]^{n \times n}_+$, it immediately follows that
\[M^* = \Phi_M(p^*,q^*,M^*) = (x + A^\top p^* - \bar p)^+ \leq (x + A^\top \bar p - \bar p)^+ = \bar M,\] 
i.e., $M^* \in [0,\bar M]$.  
Additionally, following from Assumption~\ref{assumption 2.3}, 
\[q^* = \Phi_q(p^*,q^*,M^*) = F(\sum_{i = 1}^n \gamma_i(p^*,q^*),M) \leq F(0,\bar M),\]
i.e., $q^* \in [0,F(0,\bar M)]$.

From this preceding argument, we can equivalently consider the clearing problem for consider $\bar\Phi: [0,\bar p] \times [0, F(0,\bar M)] \times [0,\bar M] \to [0,\bar p] \times [0, F(0,\bar M)] \times [0,\bar M]$ which is the restriction of $\Phi$ to this complete lattice, i.e., $\bar\Phi(p,q,M) = \Phi(p,q,M)$ for any $(p,q,M) \in [0,\bar p] \times [0, F(0,\bar M)] \times [0,\bar M]$.
According to the Tarski fixed point theorem, if $\bar\Phi$ is non-decreasing, then the set of fixed points is a complete lattice and, as a direct consequence, there exists a greatest and least clearing payment vector, vector of price, and vector of market liquidity. Thus, we need to prove that $\left(p^{1}, q^{1},M^{1}\right) \geq\left(p^{2}, q^{2},M^{2}\right)$ implies $\bar\Phi\left(p^{1}, q^{1},M^{1}\right) \geq \bar\Phi\left(p^{2}, q^{2},M^{2}\right)$. Assume $\left(p^{1}, q^{1},M^{1}\right) \geq\left(p^{2}, q^{2},M^{2}\right)$.
First, for fixed $M^{1}$, $\bar\Phi\left(p^{1}, q^{1},M^{1}\right) \geq \bar\Phi\left(p^{2}, q^{2},M^{1}\right)$ has been proved in \cite[Theorem 3.4]{feinstein2017financial}. Therefore, if $\bar\Phi\left(p^{2}, q^{2},M^{1}\right) \geq \bar\Phi\left(p^{2}, q^{2},M^{2}\right)$ for $M^{2}  \leq  M^{1}$, the result follows. Note that, as the payment $p^{2}$ and price $q^{2}$ are fixed, the liquidations $\gamma(p^{2}, q^{2})$ are also fixed. Thus, $\bar\Phi_p(p^{2}, q^{2}, M^{1}) = \bar\Phi_p(p^{2}, q^{2}, M^{2})$ and $\bar\Phi_M(p^{2}, q^{2}, M^{1}) = \bar\Phi_M(p^{2}, q^{2}, M^{2})$. Under Assumption \ref{assumption 2.3}, the inverse demand function is monotonically non-decreasing in $M$, i.e., $\bar\Phi_{q}\left(p^{2}, q^{2}, M^{1}\right) \geq  \bar\Phi_{q}\left(p^{2}, q^{2}, M^{2}\right)$. Therefore, $\bar\Phi\left(p^{2}, q^{2},M^{1}\right) \geq \bar\Phi\left(p^{2}, q^{2},M^{2}\right)$, and the proof is completed.
\end{proof}

\subsection{Counterexample to uniqueness}\label{sec:clearing-nonunique}
Though we prove the existence of greatest and least clearing solutions in Theorem~\ref{Theorem 3.2.1}, prior works on fire sales (e.g., \cite{amini2016uniqueness}) have guaranteed the uniqueness of this clearing solution under mild and natural technical assumptions on the inverse demand function. However, when explicitly introducing market liquidity to the system, we find that uniqueness is no longer guaranteed to hold even under those conditions. In this section, we construct a simple financial system that admits two distinct clearing solutions.

\begin{example}\label{ex:clearing-nonunique}
Consider a network with two banks and a society node with a single type of illiquid asset ($n=2$ and $m=1$). As depicted in Figure \ref{Fig.main1}, bank 1 holds no liquid assets, bank 2 holds $0.01$ in external liquid assets ($x_1 = 0$ and $x_2 = 0.001$). Further, both banks hold illiquid assets to the amount of  $s_1 = 2.35$ and $s_2 = 2$. Additionally, bank 2 holds 1 unit of interbank assets owed by bank 1 ($L_{12} = 1$) with no other interbank obligations in this system and both banks have an obligation of 1 unit externally to the system ($L_{10} = L_{20} = 1$).
Finally, for this construction, we consider the liquidity-adjusted linear inverse demand function of Example~\ref{example2.1} with parameters $(\mu, C) = (1,1)$ and risk aversions $\alpha_0 = \frac{1}{15},~\alpha_1 = \alpha_2 = 1$, i.e.,
\[F(\theta,M) = 1 - \frac{\theta}{15 + |\M|}\]
where $|\M| = \#\{i \; | \; M_i > 0\}$ and $\theta \geq 0$. 
Notably, the interbank network is a regular network as defined in \cite[Definition 5]{eisenberg2001systemic} and under this construction, $\theta \in [0,4.35] \mapsto \theta F(\theta,M)$ is strictly increasing for any possible market liquidity $M \in \mathbb{R}^2_+$; that is, these parameters were chosen so that the uniqueness condition of \cite{amini2016uniqueness} is satisfied.

\begin{figure}[H]
\begin{center}
\includegraphics[width=1\linewidth]{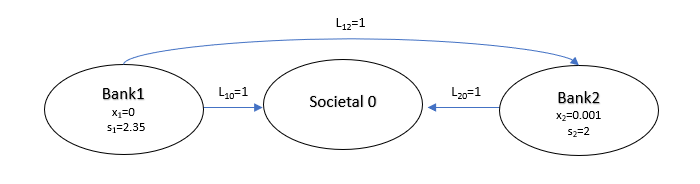}
\caption{The financial system for Example~\ref{ex:clearing-nonunique}}
\label{Fig.main1}
\end{center}
\end{figure}

Due to the construction of the liquidity-adjusted linear inverse demand function in this example, rather than considering ``fictitious'' defaults (as in Algorithm~\ref{algo}), we can fix $|\M| \in \{0,1,2\}$ and consider the clearing system of \cite{amini2016uniqueness}, i.e., under the inverse demand function $f_{|\M|}(\theta) := F(\theta,M)$ where $M$ coincides with the choice of $|\M|$. If the resulting number of market makers corresponds to the assumed $|\M|$ then this is a clearing solution, otherwise it is not. Immediately, upon applying this setup, we find that there are two distinct clearing solutions:
\begin{align*}
(p^\uparrow,q^\uparrow,M^\uparrow) &= ((2,1) \; , \; 0.854 \; , \; (0,0.001))\\
(p^\downarrow,q^\downarrow,M^\downarrow) &= ((1.98,1) \; , \; 0.843 \; , \; (0,0))
\end{align*}
where no clearing solution coincides with $|\M| = 2$.
\end{example}

\section{Case studies}\label{sec:cs}
In this section, we present two numerical case studies to study the financial implications of this extended fire sale model. 
First, in Section~\ref{sec:cs-1}, we compare the proposed joint clearing system presented within this work to the fixed liquidity setting presented in, e.g., \cite{amini2016uniqueness,greenwood2015vulnerable}.
Second, in Section~\ref{sec:cs-diversification}, we consider the impacts of varying bank portfolio constructions on systemic risk. This latter case study investigates the impact of diversification as was also undertaken in, e.g., \cite{capponi2023systemic,elliott2014financial}.
An additional case study using 2011 European Banking Authority data is provided in Appendix~\ref{sec:EBA}; the results therein follow the same stylized structure as presented within this section.

\subsection{Impacts of market maker distress}\label{sec:cs-1}
Consider a financial system with $n = 50$ banks and single ($m = 1$) illiquid asset. 
Each bank $i$ holds external assets split between its liquid portfolio $x_i \in [2,5]$ and illiquid portfolio $s_i = 4$. The liquid assets $x_i$ are varied so as to simulate different stress scenarios.
Furthermore, each bank $i$ owes $L_{i0} = 3$ external to the banking system.
Finally, the interbank assets and liabilities are randomly sampled from a uniform distribution $L_{ij} \sim U([0,1])$ for $i \neq j$ (and $L_{ii} = 0$ for any bank $i$).
The relative liabilities matrix \( A \) and the total liabilities vector \( \bar{p} \) are computed from the interbank liabilities matrix \( L \) as delineated in Section~\ref{sec:setting-payments}.

Under Assumption \ref{Assumption 2.1}, provided they have sufficient assets, banks sell the fewest assets necessary to fulfill their liquid shortfall \(s_i \wedge \left(\bar{p}_{i}-x_{i}-\sum_{j=1}^{n} a_{j i} p_{j}\right)^{+}\). That is, as in \cite{amini2016uniqueness}, the (vector) liquidation function is:
\begin{equation}
\gamma(p, q)=\frac{1}{q}\left[s \wedge \left(\bar{p}-x-A^\top p\right)^{+}\right].
\end{equation}

Finally, to complete the setting for this case study, we consider the liquidity-adjusted linear inverse demand function provided in Example \ref{example2.1} with mean and standard deviation of 1, i.e., \[F(\theta,M) = 1 - \left(\frac{1}{\alpha_0} + \sum_{i \in \M}\frac{1}{\alpha_i}\right)^{-1} \theta\]
for societal risk aversion $\alpha_0 = 0.1$ and bank risk aversion $\alpha_i = 0.1$ for every bank $i$.
This consistent choice of risk aversion allows us to isolate and study the system's equilibrium under market maker distress without the confounding effects of heterogeneous risk preferences.

For the purposes of this case study, we compare our extended clearing system with the benchmark model of \cite{amini2016uniqueness} with fixed market liquidity. Specifically, this benchmark system utilizes the linear inverse demand function $\theta \mapsto F(\theta,\vec{1})$ which corresponds to the setting without any market maker distress.\footnote{As the liquidity-adjusted linear inverse demand function only depends on the market liquidity $M$ through the set of market makers $\M$, we arbitrarily took $M = \vec{1}$ to correspond to the no market maker distress setting.}

\begin{figure}[H]
  \centering
  \begin{subfigure}[t]{0.45\linewidth}
    \includegraphics[width=\linewidth]{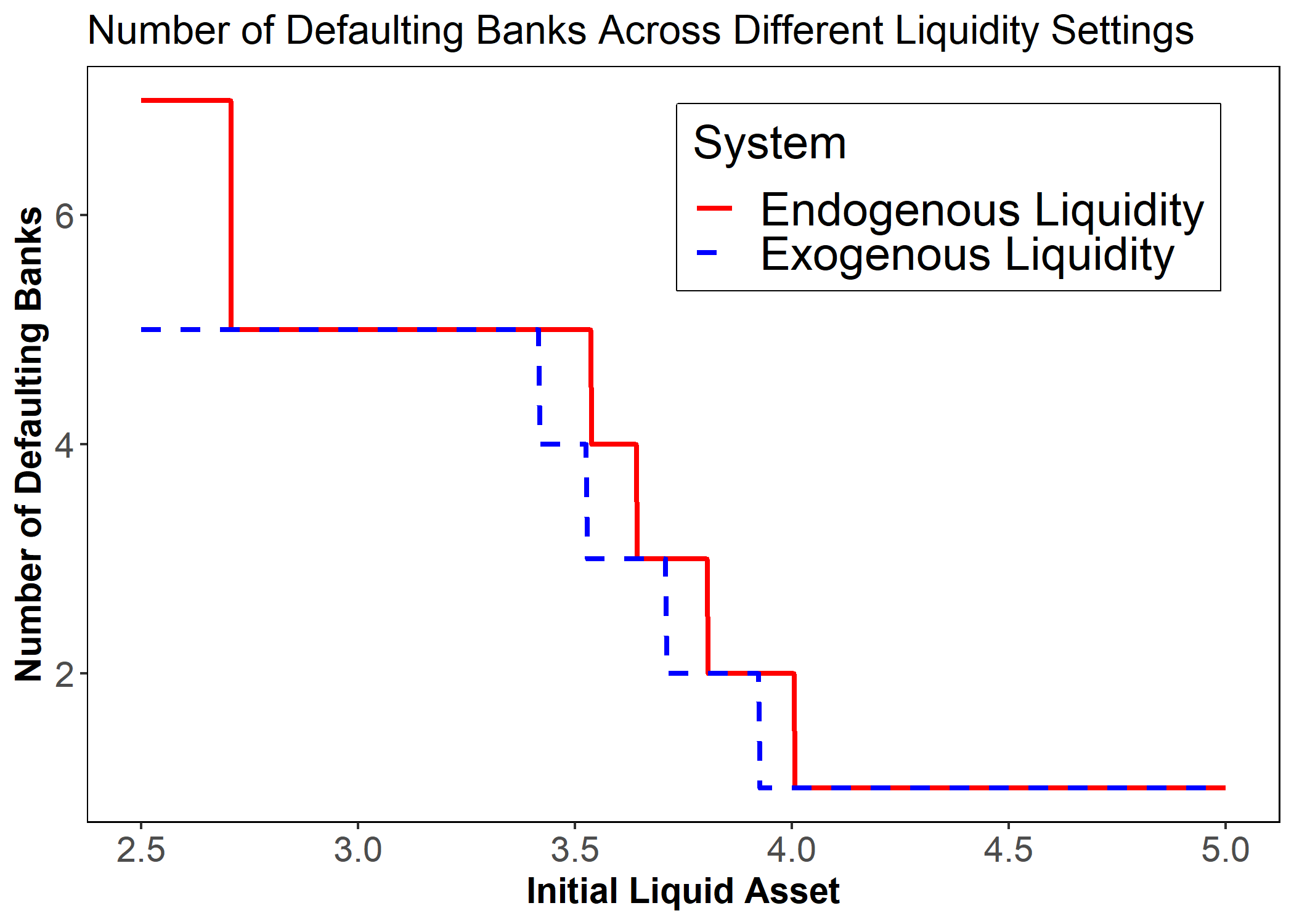}
    \caption{Defaulting Banks}
    \label{subfig:defaulting_banks}
  \end{subfigure}
  \hfill
  \begin{subfigure}[t]{0.45\linewidth}
    \includegraphics[width=\linewidth]{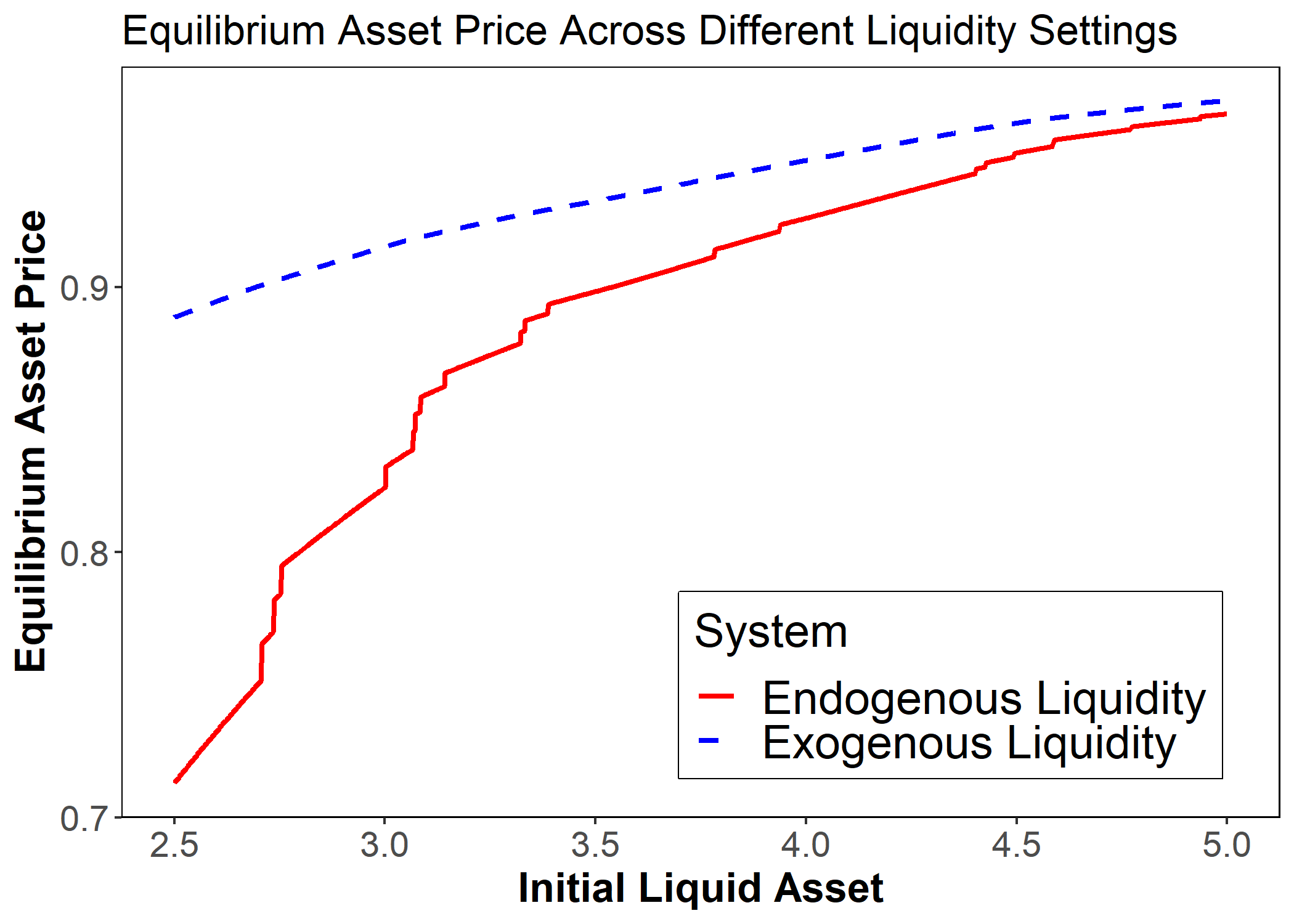}
    \caption{Equilibrium Asset Price}
    \label{subfig:asset_price}
  \end{subfigure}
  \caption{Effects of market maker distress: (a) number of defaulting banks and (b) equilibrium asset price across varying market liquidity.}
  \label{Fig.main}
\end{figure}

Figure \ref{Fig.main} provides a comprehensive view of the effects of market maker distress on the financial system under a single realization of the liability matrix. First, in Figure \ref{subfig:defaulting_banks}, we see that the number of defaulting institutions grows faster with the liquidity-adjustments than the benchmark model. Notably, for mild shocks ($x_i > 4$ for all banks $i$), both clearing systems are resilient and exhibit no defaults. However, as we escalate the financial shocks, the system with endogenous market liquidity shows a more rapid drop in resilience, resulting in a higher number of defaulting banks. Specifically, under severe financial shocks (when \( x_{i} = 2.5 \)), 7 banks default in the endogenous market liquidity system compared to only 5 in the fixed market liquidity system. 

Figure \ref{subfig:asset_price} presents the equilibrium asset prices in response to these initial financial shocks for both endogenous and exogenous market liquidity settings. In our endogenous market liquidity system, there are noticeable price jumps. These events can be attributed to shocks which trigger an additional market maker distress. Recall that, within the endogenous market maker setting, severe financial shocks not only compel banks to sell illiquid assets but also lead to distress among market makers which affects market liquidity. 
In contrast, the exogenous market liquidity setting results in a smooth mapping between shocks and the resulting asset price.

\begin{figure}[H]
  \centering
  \begin{subfigure}[t]{0.45\linewidth}
    \includegraphics[width=\linewidth]{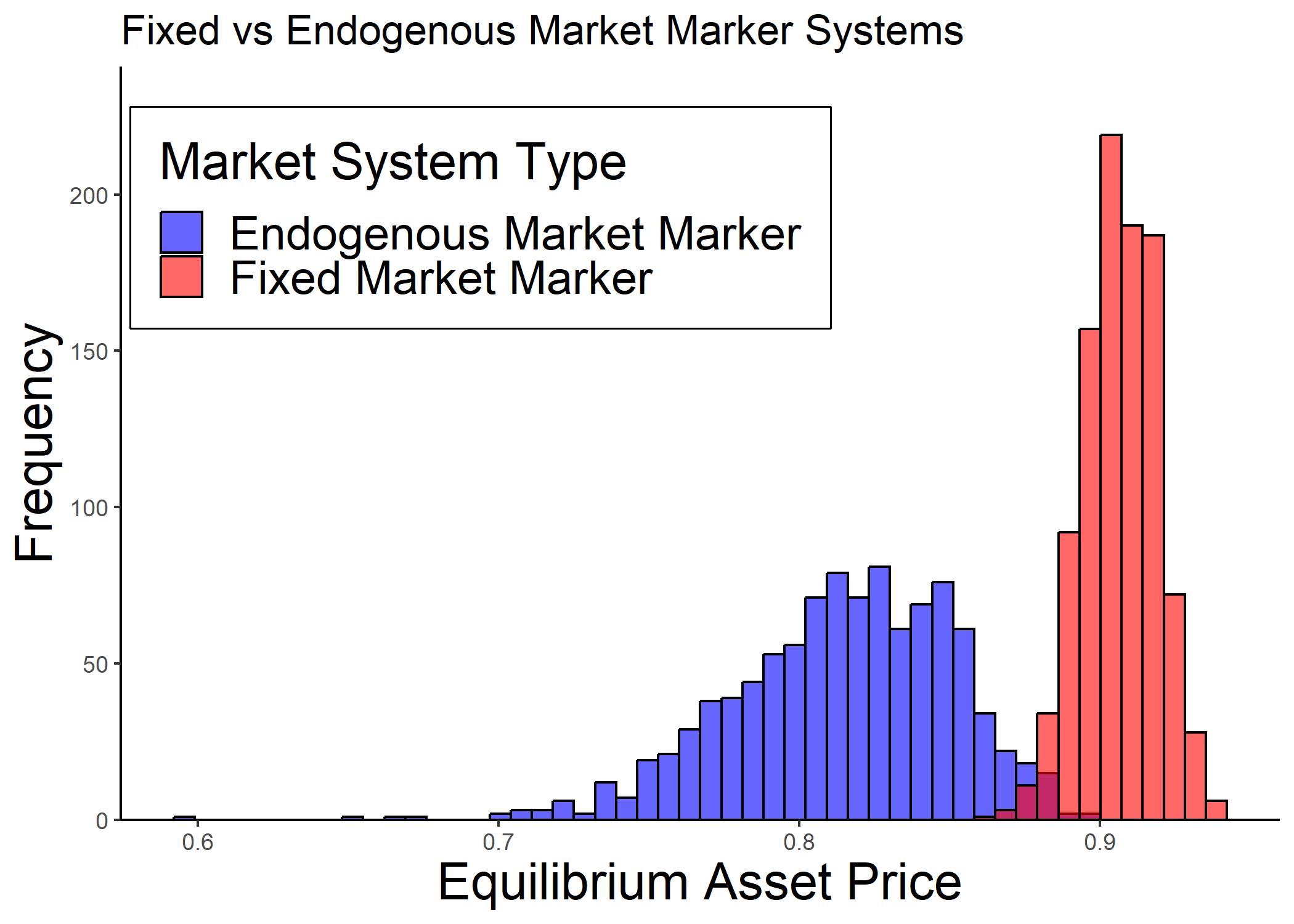}
    \caption{Fixed vs Endogenous}
    \label{subfig:Fixed vs Endogenous}
  \end{subfigure}
  \hfill
  \begin{subfigure}[t]{0.45\linewidth}
    \includegraphics[width=\linewidth]{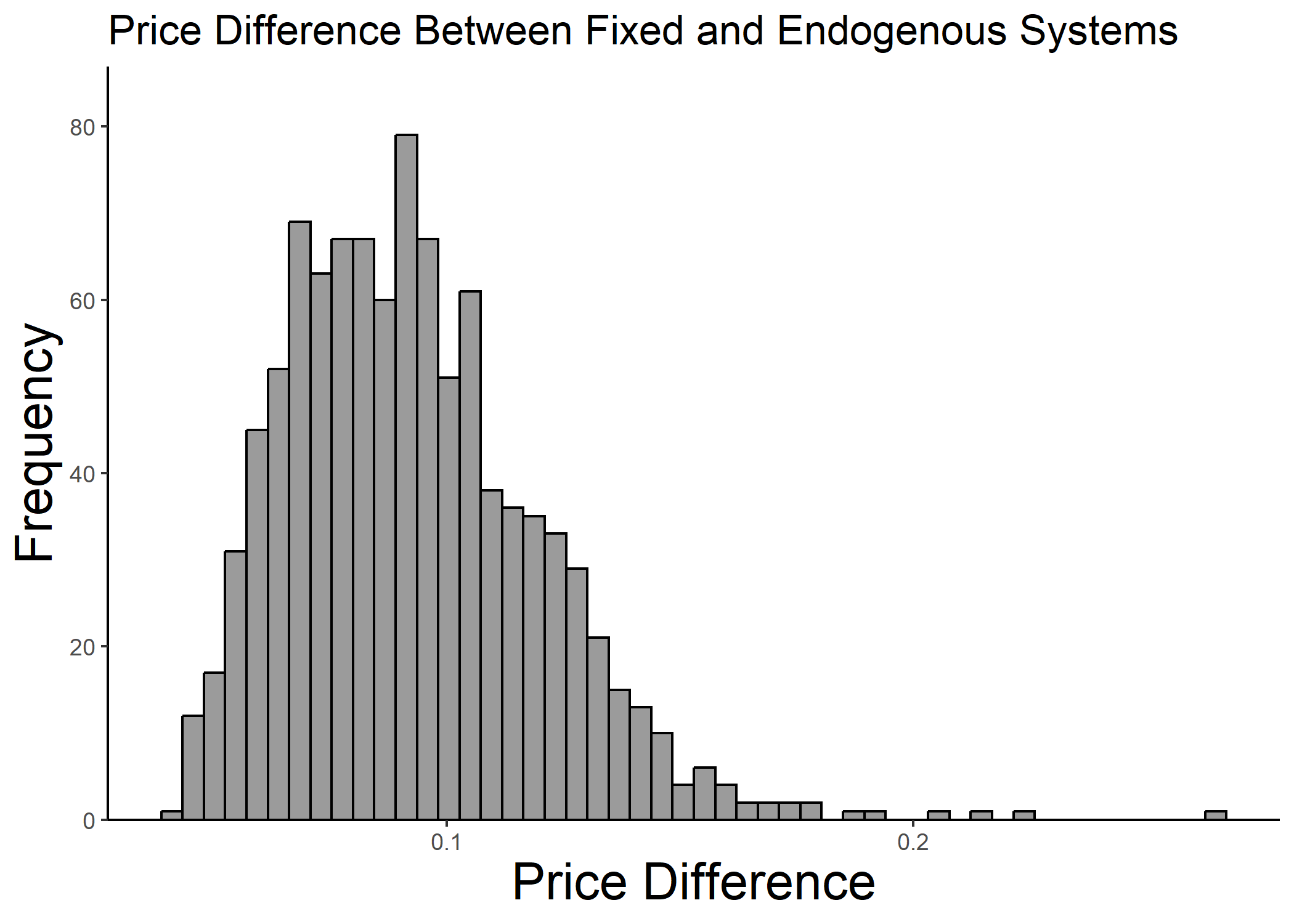}
    \caption{Price difference}
    \label{subfig:Price difference}
  \end{subfigure}
  \caption{Comparison of the equilibrium price under exogenous and endogenous liquidity settings over 1000 random networks.}
  \label{hist}
\end{figure}

To further evaluate the impact of market maker distress, we conduct simulations on 1000 randomly generated financial systems with severe financial shock (\( x = 3 \)). Figure \ref{subfig:Fixed vs Endogenous} portrays the equilibrium asset price distributions for both fixed and endogenous market marker systems. Notably, the exogenous liquidity system shows a price distribution clustered between $0.85$ and $0.95$. In contrast, the distribution for the endogenous liquidity system exhibits extreme tail behavior, with roughly $40\%$ of the cases resulting in equilibrium prices falling below $0.8$. Moreover, the endogenous system captures extreme events where the asset price drops by over $30\%$. Notably, these extreme scenarios are completely absent in the exogenous liquidity model. To highlight the differences in these clearing models, Figure \ref{subfig:Price difference} compares the equilibrium prices across individual simulations. The data reveals that the prices in the endogenous system are uniformly lower than in the exogenous system as expected. Specifically, the average equilibrium price for the endogenous system stands at $0.814$, considerably less than the $0.907$ average for the fixed system.

\begin{remark}
Recall that we considered the liquidity-adjusted linear inverse demand function within this case study. Notably, that function is discontinuous at the points that market makers enter distress. However, though it is tempting to attribute the discontinuities found in the clearing solutions to this property of the inverse demand function, we find that similar results hold for even jointly continuous liquidity-adjusted inverse demand functions. 
For instance, similar results can be found with the inverse demand function $F(\theta,M) = \exp(-\theta/\sum_{i = 1}^n M_i)$ for any $\theta \geq 0$ and $M \in \mathbb{R}^n_+$.
Therefore, we wish to highlight the fundamental nature of the discontinuities which mark points at which the distress of a single institution can lead to disproportionate impacts on the health of the financial system.
\end{remark}

\subsection{Effects of diversification}\label{sec:cs-diversification}
We now wish to consider a $n = 2$ bank system with $m = 2$ illiquid assets. We will assume that, after some shocks, bank 1 does not hold any liquid endowment, and bank 2 holds $1$ unit of liquid assets ($x_{1}=0$ and $x_{2}=1$).  
We assume that these banks share the total assets available on the market (fixed at 2 for simplicity), i.e., $s_{11} + s_{21} = 2$ and $s_{12} + s_{22} = 2$. Furthermore, assuming the initial price of both illiquid assets is 1, both banks begin the system with 2 in illiquid assets, i.e., $s_{11} + s_{12} = 2$ and $s_{21} + s_{22} = 2$.
Finally, we assume that bank 1 owes $L_{12} = 1$ to bank 2 and $L_{10} = 1.85$ to the external system; bank 2 owes $L_{21} = 1$ to bank 1 and $L_{20} = 1$ to the external system. 
Notably, the interbank network was constructed so that, if bank 1 is solvent and makes its payments in full then bank 2 does not need to liquidate any assets and, as such, acts as a market maker. However, if bank 1 is insolvent, bank 2 must liquidate assets to remain solvent.
This network is displayed in Figure \ref{Example}. 

\begin{figure}[!h]
\begin{center}
\includegraphics[width=1\linewidth]{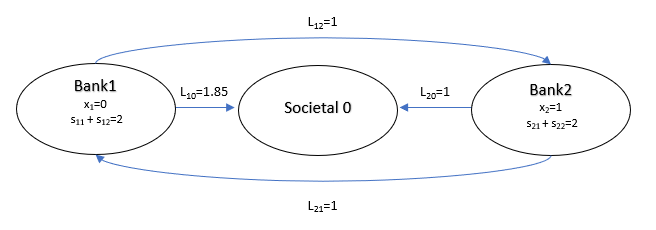}
\caption{Summary of the financial system for Section~\ref{sec:cs-diversification}}
\label{Example}
\end{center}
\end{figure}

By varying the portfolio composition of these two banks, we are able to investigate the impact of portfolio diversification on system risk in much the same as as in Example 5.4 of \cite{feinstein2020capital}. 
Due to the asset constraints described above, we parameterize the holdings by $\lambda \in[0,2]$ and set $s_{11} = s_{22} =\lambda$ and $s_{12} = s_{21} = 2-\lambda$. When $\lambda \in\{0,2\}$, the banks are holding non-overlapping portfolios, corresponding to a fully diverse system. When $\lambda=1$, the portfolios of the banks are identical, corresponding to a fully diversified system. Thus as $\lambda$ increases from 0 to 1 (or decreases from 2 to 1), the system moves from a fully diverse to a fully diversified setting. As such, we can study the effect of the diversification parameter $\lambda$ on the health of the financial system.
As in \cite{feinstein2020capital}, both banks are assumed to follow the proportional liquidation function discussed within Example~\ref{example2.2.3}.

In order to complete our discussion of the setting, we again consider the liquidity-adjusted linear inverse demand function of Example~\ref{example2.1}, i.e., 
\[F(\theta, M)=\mu - \frac{1}{10(1+|\M|)} C \theta \text{ with } \mu=1,~C = \left[\begin{array}{ll}\sigma^2 & \rho \sigma^2 \\ \rho \sigma^2 & \sigma^2\end{array}\right].\]
Implicit to this construction, all risk aversions $\alpha_{i}=0.1$ are fixed at the same level. 

To fully study the impacts of this clearing model, we will vary both the standard deviation $\sigma > 0$ and the correlation $\rho \in [0,1]$.
Note that, if $\sigma > 0$ were kept constant, the cross-price impacts would only worsen as asset correlations $\rho \in [0,1]$ increase. To counteract this effect, so as to have a more precise comparison, we will vary $\sigma$ in tandem with $\rho$ so that $(1+\rho)\sigma^2 = 1$ throughout; this maintains a constant variance of $8$ for the market portfolio $(2,2)$. In this way, increased cross-price impacts are partially offset by lower within-asset price impacts.

\begin{remark}
Due to the symmetry of the assets, we consider only the diversification parameter $\lambda \in [0,1]$ for the remainder of this example.
\end{remark}

\begin{figure}[!h]
\begin{center}
\includegraphics[width=0.6\linewidth]{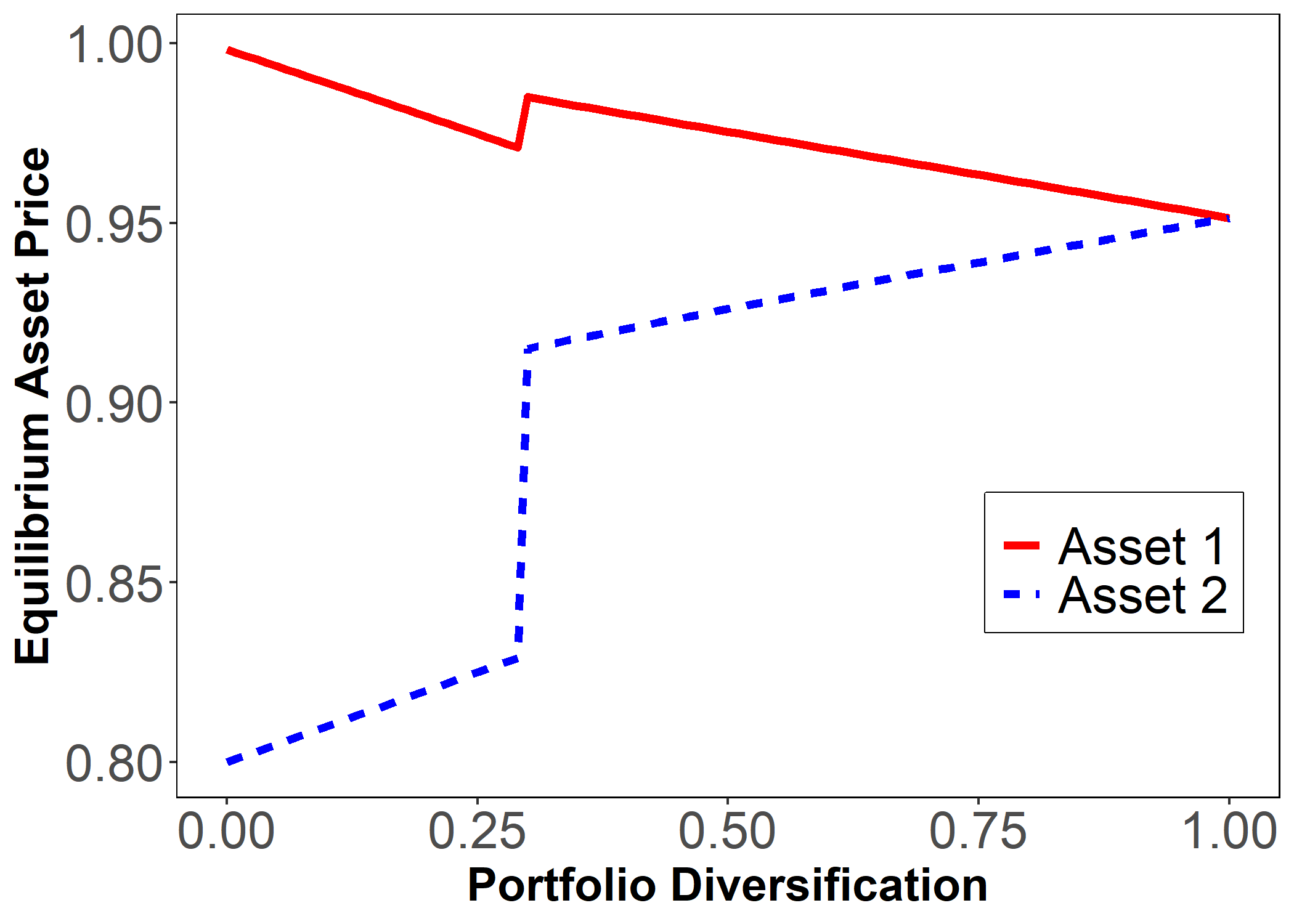}
\caption{Clearing prices with varying diversification with  $\sigma = 1$ and $\rho = 0$}
\label{cross impact 1}
\end{center}
\end{figure}

In Figure \ref{cross impact 1}, we find the influence of diversification on the clearing prices for these two assets under $\sigma = 1$ and $\rho = 0$. For highly diverse portfolios, bank 1 is forced to liquidate a large amount of asset 2, which collapses the price of asset 2 ($q_2 \leq 0.7$). Due to these price impacts, bank 1 is not able to cover its shortfall and defaults on its obligations; as such, bank 2 needs to liquidate assets to remain solvent. That is, the price impacts are self-reinforcing as neither bank is able to serve as a market maker. 
In contrast, for highly diversified portfolios, bank 1 splits its liquidations between both assets causing less price impacts to either. As such, bank 1 will be solvent and both bank 2 and the external system act as marker makers. It is this change in market makers that causes the prices to jump. 
Specifically, when the set of market makers jumps from $\{0\}$ to $\{0,2\}$, the price of asset 1 jumps from 0.97 to 0.985 and the price of asset 2 jumps from 0.83 to 0.915.

\begin{figure}[H]
\begin{center}
\includegraphics[width=0.6\linewidth]{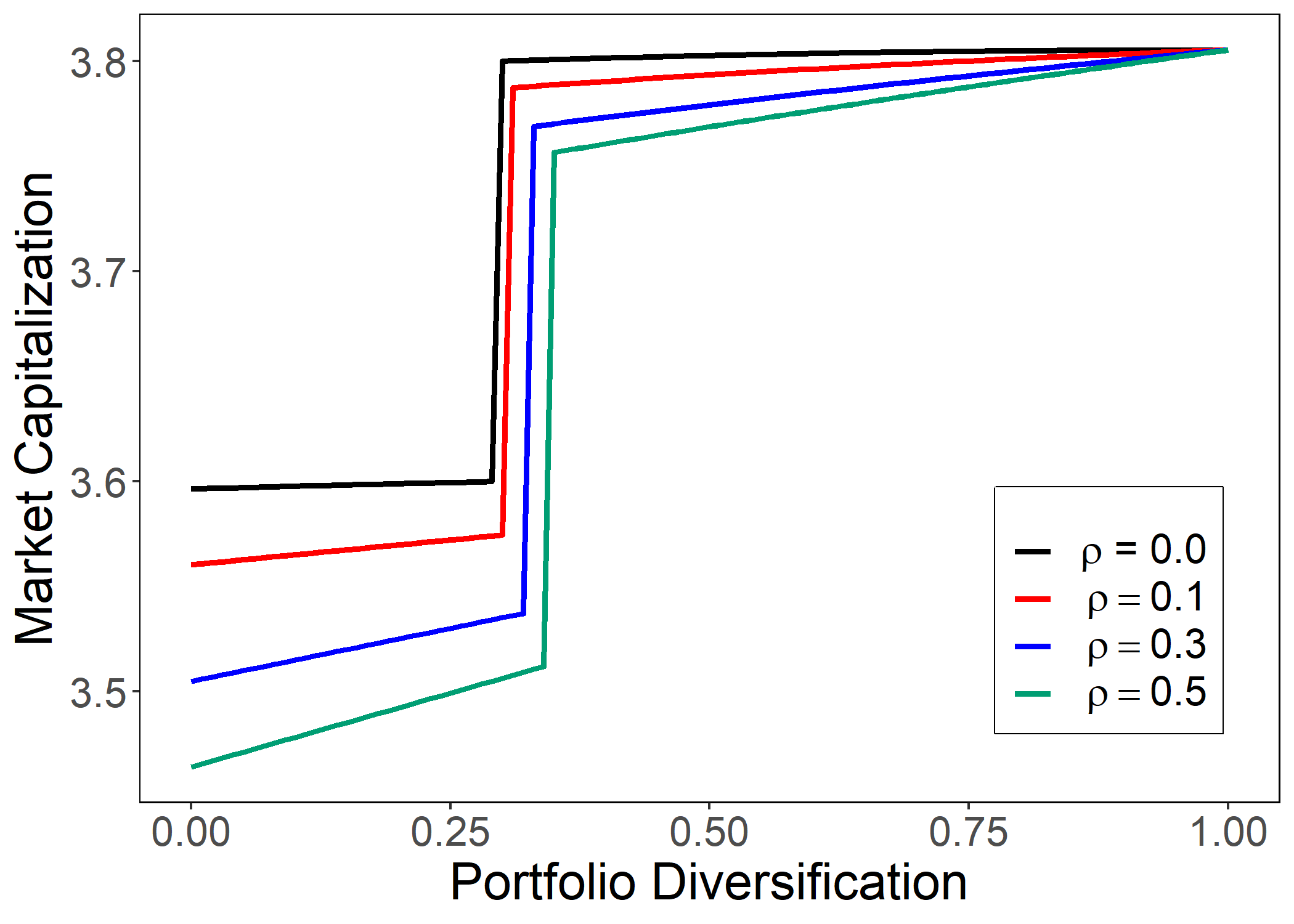}
\caption{Equilibrium market capitalization varying diversification and asset correlations.}
\label{market cap}
\end{center}
\end{figure}

In Figure \ref{market cap}, we consider the impacts of cross-price impacts through varying the correlation coefficient $\rho \in \{0,0.1,0.3,0.5\}$. Following Example \ref{example2.1}, the correlation between two assets influences the inverse demand function. Specifically, this positive correlation implies that the fire sale of one asset depreciates the price of the other. Note that the case of $\rho = 0$ is displayed in greater details within Figure \ref{cross impact 1}. 
Comparing these four correlation settings, we find that the effect of $\rho$ becomes evident on the total market capitalization (i.e., $2q_1 + 2q_2$ due to the market construction). First, the correlation influences the clearing price for a fully diverse portfolio. 
Notably, the cross-price impacts dominate the within-asset price impacts (when keeping total market volatility constant). That is, the market with low asset correlation reaches higher clearing prices and market capitalization than the market with higher correlations. 
However, for the fully diversified portfolio $\lambda = 1$, the clearing prices are equivalent (at 0.95) in all correlations always stay. 
In fact, it can easily be proven that the perfectly diversified system is independent of the asset correlations as observed herein. 
In addition to the larger price impacts, asset correlations influence the realized price jump caused by the change in the set of market markers. For $\rho=0$, the price jump occurs at $\lambda=0.3$, whereas for $\rho=0.5$, the price jump occurs at $\lambda=0.36$.

\bibliographystyle{plain}
\bibliography{ref}

\newpage
\appendix
\section{Fictitious default algorithm}\label{sec:clearing-fda}
In this section, we introduce a modified version of a fictitious default algorithm from \cite{eisenberg2001systemic} to calculate the greatest clearing solution $\left(p^{\uparrow},q^{\uparrow}, M^{\uparrow}\right)$. Briefly, the fictitious default algorithm works as follows. First, we assume that all obligations are paid in full $(p=\bar{p})$ and the prices of the illiquid assets are at the maximum level $(q=\bar{q})$ based on this assumption. We calculate the set of insolvent banks in the first iteration. Under this set of insolvent banks, we calculate the equilibrium prices of illiquid assets, the related payment vector, and market liquidity under the clearing mechanism in Section 3.1, i.e., $\left(p^*, q^*, M^*\right)=\Phi\left(p^*, q^*, M^*\right)$. With these candidate payments, prices, and market liquidity, we update the set of insolvent banks. If the set of insolvent banks does not expand, then the algorithm terminates. If the set of insolvent banks changes, we iterate the price of illiquid assets, the related payment vector, and market liquidity accordingly. This process is repeated until the set of insolvent banks no longer changes.

\begin{algorithm} \label{algo}
Initialize $k=0$. $p^{k}=\bar{p}$, $q^{k}=\bar{q}$, and $M^{k}=\bar{M}$.
\begin{enumerate}[label=(\roman*)]
\item Increment $k=k+1$;
\item Denote the set of insolvent banks: $D^k:=\{i \in\{1,2, \ldots, n\} \mid x_i+\sum_{l=1}^m s_{i l} q_l^{k-1}+\sum_{j=1}^n a_{j i} p_j^{k-1}-\bar{p}_i < 0 \}$;
\item If $k \geq 2$ and $D^k=D^{k-1}$ then terminate;

\item Define the matrix $\Lambda^k \in\{0,1\}^{n \times n}$ so that
$$
\Lambda^k_{i j}= \begin{cases}1 & \text { if } i=j \in D^k \\ 0 & \text { else }\end{cases};
$$
\item $p^k=\hat{p}, q^k=\hat{q}$, and $M^k=\hat{M}$ \text { are the maximal solution of the following fixed point problem }\\
$\hat{p}=(I-\Lambda^k) \bar{p}+\Lambda^k\left(x+S \hat{q}+ A^{\top} \hat{p}\right)$\\
$\hat{q}=F(\sum_{i \in D^k} s_i+\sum_{i \notin D^k}\left[s_i \wedge \gamma_i(\hat{p}, \hat{q})\right], \hat{M})$\\
$\hat{M}=\left(x+A^{\top} \hat{p}-\bar{p}\right)^{+}$

\item Return to (i);
\end{enumerate}

\end{algorithm}

\begin{proposition}
The maximal clearing solution can be found by the algorithm in at most $n$ iterations.
\end{proposition}

\begin{proof}
The convergence of this algorithm to the greatest clearing solution $\left(p^{\uparrow},q^{\uparrow}, M^{\uparrow}\right)$ follows from the logic of the fictitious default algorithm in \cite{eisenberg2001systemic}. The set of insolvent banks will expand for each iteration, i.e., $ D^{k+1} \subseteq D^{k}$. For $k = 1$, $(p^{1}, q^{1}) \leq (p^{0}, q^{0}) = (\bar{p}, \bar{q}) $. For $k \geq 1$, We assume that $(p^{k}, q^{k}) \leq (p^{k-1}, q^{k-1})$. Solving the fixed point problem in a fictitious default algorithm, we have $(p^{k+1}, q^{k+1}) \leq (p^{k}, q^{k})$. Therefore, the algorithm converges in at most $n$ iterations as there are only $n$ financial institutions.
\end{proof}

\section{European Banking Authority case study}\label{sec:EBA}

We wish to augment the small numerical experiments conducted within the main body of the text with a large numerical case study calibrated to financial data. Herein, we seek to see if the results found in the main text also hold for a more realistic financial network. Following the calibration utilized in \cite{feinstein2018sensitivity}, we utilize detailed bank-by-bank data for $n = 87$ banks across 25 EEA/EU countries from the European Banking Authority's EU-wide Transparency Exercise. Specifically, we employ the 2011 EBA dataset, which has been used in previous studies on systemic risk (see, e.g., \cite{gandy2017bayesian,chen2016optimization}).

For the purposes of this case study we consider the case in which only a single illiquid asset $m = 1$ exists. As only the total external assets, and not their composition, are known, we simulate this system by varying the proportion of those holdings that are liquid and illiquid at all institutions. For the purposes herein, we vary that ratio from 90\% to 100\% liquid (equivalently 10\% to 0\% illiquid). Using the liquidity-adjusted linear inverse demand function (Example~\ref{example2.1}), we set all risk aversions $\alpha_i = 5 \times 10^{-7}$ and $\alpha_0 = 5 \times 10^{-7}$. Though we only consider this singular risk aversion, all results appear comparable with other choices of risk aversions as well.

\begin{figure}[ht]
\begin{center}
\includegraphics[width=0.5\linewidth]{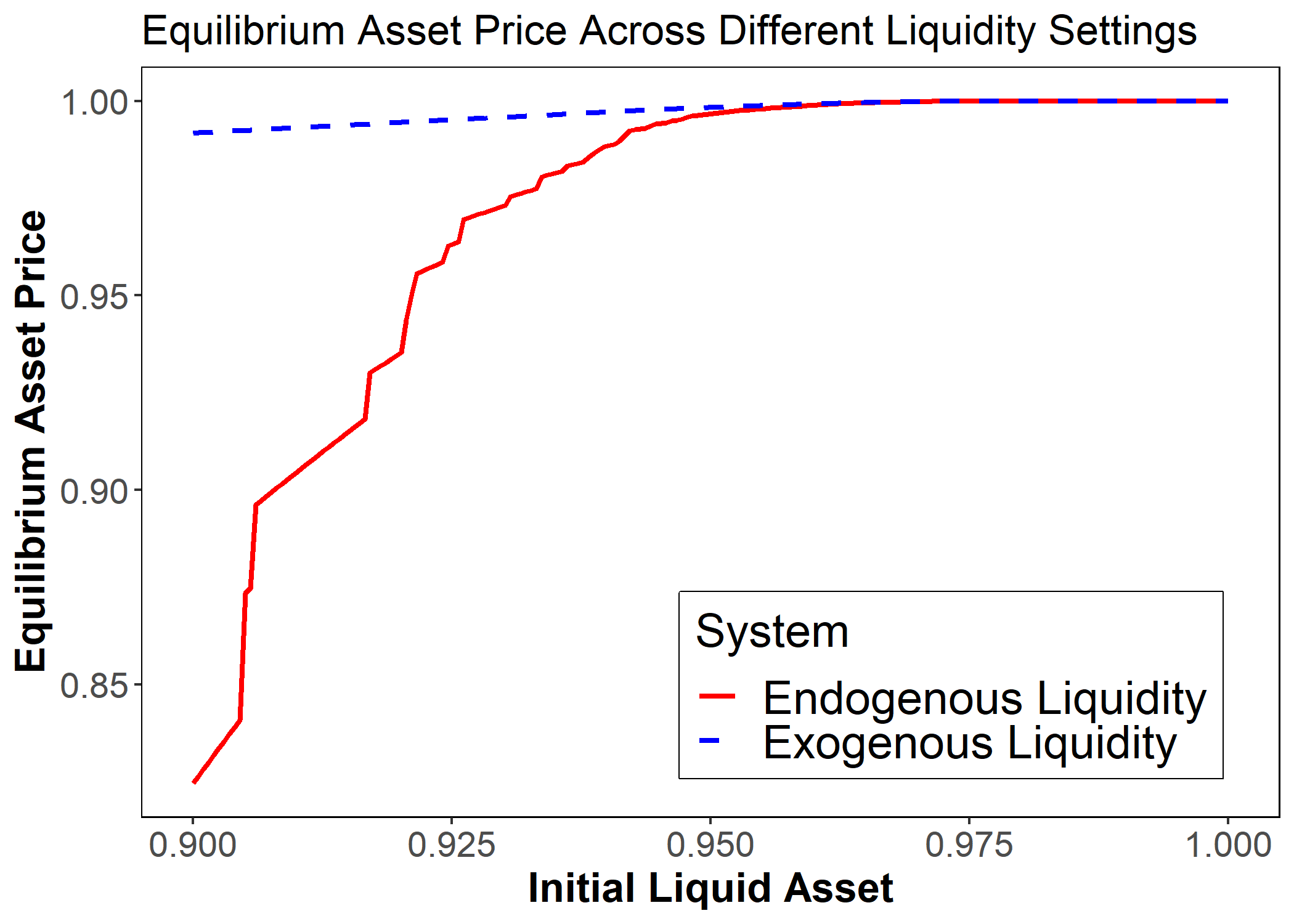}
\caption{Effects of initial asset holdings on equilibrium asset prices on a system calibrated to European banking data.}
\label{EBA}
\end{center}
\end{figure}

Figure~\ref{EBA} compares our proposed liquidity-adjusted system to the fixed liquidity setting used in prior works (see, e.g., \cite{amini2016uniqueness,greenwood2015vulnerable}). In this European banking system, we observe noticeable price jumps in our endogenous market liquidity model. These jumps can be attributed to shocks triggering additional market maker distress, demonstrating that the magnitude of impact indicated by our stylized simulations of Section~\ref{sec:cs} also holds for realistic networks.

While our model provides valuable insights into the interplay between market liquidity and systemic risk, it is essential to acknowledge some potential limitations in this case study. First, for simplicity, we assume that market makers have homogeneous risk aversions and distribution of asset holdings. However, in practice, market makers may have heterogeneous characteristics, which could lead to more complex market behaviors. Second, our model considers a single period of trading and does not account for the potential long-term effects of market maker distress on market liquidity and asset prices. These limitations present several avenues for future research, such as exploring the impact of incomplete or asymmetric information, investigating the effects of heterogeneous market maker characteristics, and extending the model to a multi-period setting to capture long-term market dynamics.

\end{document}